\documentclass[10pt,a4paper]{amsart}
\usepackage[margin=20mm]{geometry}

\usepackage{amssymb,amsmath}
\usepackage{amsthm}
\usepackage{bbm}
\usepackage{stmaryrd}
\usepackage{hyperref}
\usepackage[usenames,dvipsnames]{xcolor}
\usepackage{color}

\input xy
\xyoption{all} 

\usepackage{lineno}

\numberwithin{equation}{section}

\newtheorem{theorem}{Theorem}[section]
\newtheorem{corollary}[theorem]{Corollary}

\newtheorem{proposition}[theorem]{Proposition}

\theoremstyle{definition}
\newtheorem{definition}[theorem]{Definition}
\newtheorem{remark}[theorem]{Remark}
\newtheorem{example}[theorem]{Example}

\newcommand{\Id}{\mathbbmss{1}}

\newcommand{\rmi}{ \textnormal{i}}

\DeclareMathOperator{\Span}{Span}

\font\black=cmbx10 \font\sblack=cmbx7 \font\ssblack=cmbx5 \font\blackital=cmmib10  \skewchar\blackital='177
\font\sblackital=cmmib7 \skewchar\sblackital='177 \font\ssblackital=cmmib5 \skewchar\ssblackital='177
\font\sanss=cmss10 \font\ssanss=cmss8 
\font\sssanss=cmss8 scaled 600 \font\blackboard=msbm10 \font\sblackboard=msbm7 \font\ssblackboard=msbm5
\font\caligr=eusm10 \font\scaligr=eusm7 \font\sscaligr=eusm5  \font\fraktur=eufm10
\font\sfraktur=eufm7 \font\ssfraktur=eufm5 
\font\bsymb=cmsy10 scaled\magstep2
\def\all#1{\setbox0=\hbox{\lower1.5pt\hbox{\bsymb
       \char"38}}\setbox1=\hbox{$_{#1}$} \box0\lower2pt\box1\;}
\def\exi#1{\setbox0=\hbox{\lower1.5pt\hbox{\bsymb \char"39}}
       \setbox1=\hbox{$_{#1}$} \box0\lower2pt\box1\;}

\def\tx#1{{\fam0\relax#1}}

\newfam\bifam
\textfont\bifam=\blackital \scriptfont\bifam=\sblackital \scriptscriptfont\bifam=\ssblackital

\newfam\blfam
\textfont\blfam=\black \scriptfont\blfam=\sblack \scriptscriptfont\blfam=\ssblack

\newfam\bbfam
\textfont\bbfam=\blackboard \scriptfont\bbfam=\sblackboard \scriptscriptfont\bbfam=\ssblackboard

\newfam\ssfam
\textfont\ssfam=\sanss \scriptfont\ssfam=\ssanss \scriptscriptfont\ssfam=\sssanss

\newfam\clfam
\textfont\clfam=\caligr \scriptfont\clfam=\scaligr \scriptscriptfont\clfam=\sscaligr

\newfam\frfam
\textfont\frfam=\fraktur \scriptfont\frfam=\sfraktur \scriptscriptfont\frfam=\ssfraktur

\def\hpb#1{\setbox0=\hbox{${#1}$}
    \copy0 \kern-\wd0 \kern.2pt \box0}
\def\vpb#1{\setbox0=\hbox{${#1}$}
    \copy0 \kern-\wd0 \raise.08pt \box0}

\def\pmb#1{\setbox0\hbox{${#1}$} \copy0 \kern-\wd0 \kern.2pt \box0}
\def\pmbb#1{\setbox0\hbox{${#1}$} \copy0 \kern-\wd0
      \kern.2pt \copy0 \kern-\wd0 \kern.2pt \box0}
\def\pmbbb#1{\setbox0\hbox{${#1}$} \copy0 \kern-\wd0
      \kern.2pt \copy0 \kern-\wd0 \kern.2pt
    \copy0 \kern-\wd0 \kern.2pt \box0}
\def\pmxb#1{\setbox0\hbox{${#1}$} \copy0 \kern-\wd0
      \kern.2pt \copy0 \kern-\wd0 \kern.2pt
      \copy0 \kern-\wd0 \kern.2pt \copy0 \kern-\wd0 \kern.2pt \box0}
\def\pmxbb#1{\setbox0\hbox{${#1}$} \copy0 \kern-\wd0 \kern.2pt
      \copy0 \kern-\wd0 \kern.2pt
      \copy0 \kern-\wd0 \kern.2pt \copy0 \kern-\wd0 \kern.2pt
      \copy0 \kern-\wd0 \kern.2pt \box0}

\mathchardef\za="710B  
\mathchardef\zb="710C  
\mathchardef\zg="710D  
\mathchardef\zd="710E  
\mathchardef\zve="710F 
\mathchardef\zz="7110  
\mathchardef\zh="7111  
\mathchardef\zvy="7112 
\mathchardef\zi="7113  
\mathchardef\zk="7114  
\mathchardef\zl="7115  
\mathchardef\zm="7116  
\mathchardef\zn="7117  
\mathchardef\zx="7118  
\mathchardef\zp="7119  
\mathchardef\zr="711A  
\mathchardef\zs="711B  
\mathchardef\zt="711C  
\mathchardef\zu="711D  
\mathchardef\zvf="711E 
\mathchardef\zq="711F  
\mathchardef\zc="7120  
\mathchardef\zw="7121  
\mathchardef\ze="7122  
\mathchardef\zy="7123  
\mathchardef\zf="7124  
\mathchardef\zvr="7125 
\mathchardef\zvs="7126 
\mathchardef\zf="7127  
\mathchardef\zG="7000  
\mathchardef\zD="7001  
\mathchardef\zY="7002  
\mathchardef\zL="7003  
\mathchardef\zX="7004  
\mathchardef\zP="7005  
\mathchardef\zS="7006  
\mathchardef\zU="7007  
\mathchardef\zF="7008  
\mathchardef\zW="700A  
\mathchardef\zC="7009  

\newcommand{\be}{\begin{equation}}
\newcommand{\ee}{\end{equation}}

\newcommand{\bea}{\begin{eqnarray}}
\newcommand{\eea}{\end{eqnarray}}

\def\*{{\textstyle *}}
\newcommand{\R}{{\mathbb R}}

\newcommand{\C}{{\mathbb C}}

\newcommand{\s}{{\textstyle *}}

\def\xi{\tx{i}}

\def\s*{{\scriptstyle *}}

\newcommand{\beas}{\begin{eqnarray*}}
\newcommand{\eeas}{\end{eqnarray*}}

\author{Andrew James Bruce}
   \address{Department of Mathematics,
The Computational Foundry,
Swansea University Bay Campus,
Fabian Way,
Swansea, SA1 8EN, United Kingdon}
   \email{andrewjamesbruce@googlemail.com}
\begin{document}
\date{\today}
\title{Semiheaps and Ternary Algebras in Quantum Mechanics Revisited}

\begin{abstract}
	We re-examine the appearance of semiheaps and (para-associative) ternary algebras in quantum mechanics. In particular, we review the construction of a semiheap on a Hilbert space and the set of bounded operators on a Hilbert space. The new aspect of this work is a discussion of how symmetries of a quantum system induce homomorphisms of the relevant semiheaps and ternary algebras.  \\
	\smallskip


\end{abstract}
\maketitle
\vspace{-10pt}
\section{Introduction}
Heaps were introduced by Pr\"{u}fer \cite{Prufer:1924} and  Baer \cite{Baer:1929} as a set equipped with a ternary operation satisfying simple axioms. One can think of a heap as a group in which the identity element has been forgotten. Indeed, these axioms are satisfied in a group  if we define the ternary operation as $(a,b,c) \mapsto ab^{-1}c$.  For example, given a vector space or more generally an affine space, we can construct a heap operation as $(u,v,w) \mapsto u - v+w$.  Conversely, by selecting any element in a heap, one can reduce the ternary operation to a  group operation, such that the chosen element is the identity element. \par 
There is a slightly weaker notion of a \emph{semiheap}. A semiheap is a non-empty set $H$, equipped with a ternary operation $[a,b,c] \in H$ that satisfies the \emph{para-associative law} 
$$\big[ [a,b,c] , d,e \big] = \big [ a,[d,c,b],e\big] =  \big[ a,b,[c , d,e] \big]\, ,$$
for all $a,b,c,d$ and $e \in H$. A semiheap is a \emph{heap} when all its elements are \emph{biunitary}, i.e., $[a,b,b] =a$ and  $[b,b,a] = a$, for all $a$ and $b \in H$. This condition is also referred to as the Mal’cev identities. A \emph{homomorphism of semiheaps} $\phi : (H, [-,-,-]) \rightarrow (H',[-,-,-]')$ is a map $\phi : H \rightarrow H'$ such that $\phi\big([a,b,c]\big) = [\phi(a), \phi(b), \phi(c)]'$.  For more details about heaps and related structures the reader my consult Hollings \& Lawson \cite{Hollings:2017} and/or Brzeziński \cite{Brzezinski:2020}. \par 
In this paper, we re-examine the natural occurrences of semiheaps in the formalism of standard non-relativistic quantum mechanics. The semiheaps explored here have appeared scattered in the mathematics literature under different names. However, there seems to be almost nothing written with physicists and quantum mechanics in mind. The exception here is  Kerner (see \cite{Kerner:2008}), who refers to ``2nd type associativity'' or ``B-associativity'', this is precisely the above para-associativity law. \par 
In the setting of quantum mechanics, we do not just have a semiheap but also a vector space structure. A $\mathbb{C}$-vector space with a ternary product that is linear in the first and third arguments, and conjugate linear in the second argument, we will refer to as a \emph{ternary algebra} (see \cite{Abramov:2009,Bazunova:2004,Kerner:2018,Michor:1996}).  If, in addition, the ternary product is para-associative, so defines a semiheap on the underlying set, then we speak of a  \emph{para-associative ternary algebra}. We will only deal with the para-associative case past this point. A homomorphism of para-associative ternary algebras is a linear map that is simultaneously a homomorphism of semiheaps.\par 
We review the construction of semiheaps and ternary algebras on a Hilbert space and on the $*$-algebra of bounded operators on the said Hilbert space. While these constructions are not new, they are not well-known within the context of quantum mechanics. The new aspect of this work is a discussion of symmetries of quantum systems and how they induce semiheap and, in turn, ternary algebra homomorphisms. Generalised derivations of the ternary algebras are also discussed.  We will focus on algebraic aspects of the theory and not address topological issues. \par 
Rather generally, ternary operations and relations have a long history in physics. As examples, we have Nambu brackets (1973; \cite{Nambu:1973}), the Yang-Baxter equation  (1967, 1972; \cite{Baxter:1972,Yang:1967}) and the BLG model of M2-branes (2007, 2009; \cite{Bagger:2007,Gustavsson:2009}). A review of n-ary generalisations of Lie algebras and their physical applications can be found in \cite{Azcarraga:2010}. We also mention that $L_\infty$-algebras (c.f. \cite{Lada:1993}) have found a wealth of applications in physics, notably through the BV-formalism of gauge theories.  \par 
It is also curious to note that, within the standard model, the number three constantly appears. Specifically, there are three generations of quarks, three generations of leptons,  three fundamental forces (gravity is not included and is different), and three quarks are needed to make a baryon. Alongside this, there are three spatial dimensions and three fundamental inversions - charge (C), parity (P) and time (T). It is only the combination of CPT that is respected in all interactions. It is not known how, or indeed if, these threes are related. 
\section{Semiheaps associated with Hilbert spaces}
\subsection{The semiheap  and ternary algebra of a Hilbert space}
Given a vector space, there is no obvious way to multiply two vectors together and obtain another vector in the same space. However, if the vector space comes equipped with an inner product, then we can multiply three vectors together in a canonical way to obtain another vector. For the case at hand, we will restrict attention to (complex) Hilbert spaces as found in quantum mechanics. We will employ Dirac's notation throughout this paper. We will denote by $\mathcal{H}$ both a Hilbert space and its underlying set, the context should be clear. To emphasise the linear structure we will write $(\mathcal{H}, +)$.
\begin{definition}\label{Def:TerProdVec}
	Let $\mathcal{H}$ be a Hilbert space, the \emph{vector ternary product} $[-,-,-] : \mathcal{H}\times \mathcal{H} \times \mathcal{H} \longrightarrow \mathcal{H}$ is defined as
	$$[| \psi_1 \rangle  , | \psi_2 \rangle, | \psi_3 \rangle] := | \psi_1 \rangle \langle \psi_2 | \psi_3 \rangle\,.$$
\end{definition}
Recall that the norm of a vector is defined as $|| \, |\psi \rangle \, || :=\sqrt{ \langle \psi | \psi \rangle}$. It is then immediately clear that  $|| \,[| \psi_1 \rangle  , | \psi_2 \rangle, | \psi_3 \rangle]\, || = |\langle \psi_2 | \psi_3 \rangle | \, || ~ |\psi_1 \rangle \, ||  \leq || ~ |\psi_3 \rangle \, || \, || ~ |\psi_2 \rangle \, ||\, || ~ |\psi_1 \rangle \, ||  $ via the Cauchy--Schwarz inequality. \par  
The following proposition is evident.
\begin{proposition}\label{Prop:LinTerProdVec}
	Let $\mathcal{H}$ be a Hilbert space. Then the vector ternary product, see Definition \ref{Def:TerProdVec}, is linear with respect to the first and third arguments, and conjugate linear with respect to the second entry, i.e.,
	\begin{align*}
	[| \psi_1 \rangle + c_1 \, |\psi'_1 \rangle, |\psi_2 \rangle , |\psi_3 \rangle ] &= [| \psi_1 \rangle, |\psi_2 \rangle , |\psi_3 \rangle ]+  c_1 \,[ |\psi'_1 \rangle, |\psi_2 \rangle , |\psi_3 \rangle ]\, ,\\
	[| \psi_1 \rangle , |\psi_2 \rangle+ c_2 \, |\psi'_2 \rangle , |\psi_3 \rangle ] &= [| \psi_1 \rangle, |\psi_2 \rangle , |\psi_3 \rangle ]+  c_2^* \,[ |\psi_1 \rangle, |\psi'_2 \rangle , |\psi_3 \rangle ]\, ,\\
	[| \psi_1 \rangle ,  |\psi_2 \rangle , |\psi_3 \rangle + c_3 \, |\psi'_3 \rangle ] &= [| \psi_1 \rangle, |\psi_2 \rangle , |\psi_3 \rangle ]+  c_3 \,[ |\psi_1 \rangle, |\psi_2 \rangle , |\psi'_3 \rangle ]\, ,
	\end{align*}
	for all $|\psi_1 \rangle,|\psi_2 \rangle,|\psi_3 \rangle \in \mathcal{H}$ and $c_1, c_2, c_3 \in \C$.
\end{proposition}
Thus, the linear structure and the vector ternary product are compatible in the above sense. Moving on to the generalised notion of associativity we have the following theorem.
\begin{theorem}\label{Trm:ParAssTerProdVec}
	The vector ternary product on a Hilbert space $\mathcal{H}$, see Definition \ref{Def:TerProdVec}, satisfies the para-associative law
	$$\big [[| \psi_1 \rangle  , | \psi_2 \rangle, | \psi_3 \rangle], | \psi_4 \rangle  , | \psi_5 \rangle \big ] = \big [ | \psi_1 \rangle  , [| \psi_4 \rangle, | \psi_3 \rangle, | \psi_2 \rangle ] , | \psi_5 \rangle \big]= \big [| \psi_1 \rangle  , | \psi_2 \rangle,[ | \psi_3 \rangle, | \psi_4 \rangle  , | \psi_5 \rangle] \big ] \, ,$$ 
	for all $|\psi_1 \rangle,|\psi_2 \rangle,|\psi_3 \rangle \in \mathcal{H}$. In other words, $\big(\mathcal{H}, [-,-,-] \big )$ is a semiheap.
\end{theorem}
\begin{proof}
	This follows via direct computation. \\
	\begin{enumerate}
		\setlength\itemsep{1em}
		\item $\big [[| \psi_1 \rangle  , | \psi_2 \rangle, | \psi_3 \rangle], | \psi_4 \rangle  , | \psi_5 \rangle \big ] =  [|\psi_1\rangle \langle \psi_2 |\psi_3\rangle, |\psi_4\rangle , |\psi_5\rangle] = | \psi_1\rangle \langle \psi_2 | \psi_3 \rangle \langle \psi_4 |\psi_5 \rangle$. \label{ProTmA}
		\item $\big [ | \psi_1 \rangle  , [| \psi_4 \rangle, | \psi_3 \rangle, | \psi_2 \rangle ] , | \psi_5 \rangle \big] = [|\psi_1\rangle , |\psi_4 \rangle \langle \psi_3 |\psi_2\rangle, |\psi_5 \rangle] =  | \psi_1\rangle \langle \psi_2 | \psi_3 \rangle \langle \psi_4 |\psi_5 \rangle$.\label{ProTmB}
		\item $\big [| \psi_1 \rangle  , | \psi_2 \rangle,[ | \psi_3 \rangle, | \psi_4 \rangle  , | \psi_5 \rangle] \big ] = [|\psi_1 \rangle , |\psi_2 \rangle , |\psi_3 \rangle \langle \psi_4 |\psi_5 \rangle] =| \psi_1\rangle \langle \psi_2 | \psi_3 \rangle \langle \psi_4 |\psi_5 \rangle$. \label{ProTmC}
	\end{enumerate}
	\smallskip
	Clearly, \eqref{ProTmA} = \eqref{ProTmB} = \eqref{ProTmC}.
\end{proof}
Note that we do not have a heap. Specifically, the Mal’cev identities 
$$[a, b, b]=a, \qquad \textnormal{and} \qquad [b,b ,a] =a\,,$$
are not, in general, satisfied.  Explicitly, we see that
$$[|\psi_1 \rangle , |\psi \rangle , |\psi\rangle] = | \psi_1 \rangle \langle \psi | \psi \rangle\,.$$
Thus, if $|\psi\rangle$ is normalised, i.e., $\langle \psi | \psi \rangle = 1 $, then $|\psi\rangle$ is \emph{right unitary}. That is  
$$[|\psi_1 \rangle , |\psi \rangle , |\psi\rangle] = |\psi_1 \rangle\,.$$
Again, assuming that $|\psi\rangle$ is normalised, $\mathrm{P}_{|\psi\rangle} := |\psi\rangle \langle \psi|$ projects an arbitrary vector onto $|\psi \rangle$. Thus,
$$[|\psi \rangle , |\psi\rangle , | \psi_3 \rangle] =  \mathrm{P}_{|\psi\rangle} \big(|\psi_3\rangle \big)\,.$$ 
It is clear from the definition of the vector ternary product that
$$[\mathbf{0}, |\psi_2 \rangle ,|\psi_3\rangle ] = [ |\psi_1 \rangle ,\mathbf{0},|\psi_3\rangle ] = [ |\psi_1 \rangle , |\psi_2 \rangle,\mathbf{0} ] =  \mathbf{0}\,,$$
where $\mathbf{0}\in  \mathcal{H}$ is the zero vector.\par 
From Proposition \ref{Prop:LinTerProdVec}, Theorem \ref{Trm:ParAssTerProdVec} and the above discussion we see that a Hilbert space naturally comes with the structure of a ternary algebra in which the ternary product defines a semiheap (see \cite{Bazunova:2004} for further generalities on ternary algebras). Note that we have conjugate linearity in the second argument of the product rather than linearity. 
\begin{definition}
	Let $\mathcal{H}$ be a Hilbert space. Then the ternary algebra $(\mathcal{H}, +, [-,-,-])$ defined via Proposition \ref{Prop:LinTerProdVec} and Theorem \ref{Trm:ParAssTerProdVec} is referred to as the \emph{vector ternary algebra}.
\end{definition}
\begin{example}\label{Exp:ComplexLine}
	Consider the complex line $\C$ and define the inner product as $\langle z_1 , z_2 \rangle = \bar{z}_1 z_2$ for arbitrary complex numbers $z_1$ and $z_2$. Then the vector ternary product  is given by 
	$$[z_1,z_2, z_2] = z_1 \bar{z}_2 z_3\,.$$
	Thus, the complex line is a ternary algebra over itself. 
\end{example}
\begin{example}\label{Exp:SpinSpace}
	The Hilbert space we consider is finite-dimensional and given by the span of two orthonormal vectors ``spin up'' and ``spin down''
	$$\mathcal{H} = \Span_\C \big \{|\uparrow \rangle, ~ |\downarrow \rangle\big \} \cong \C^2\,.$$
	The non-zero vector ternary product of the basis elements are
	\begin{align*}
	& [|\uparrow \rangle,|\uparrow \rangle,|\uparrow \rangle] = |\uparrow \rangle\,,&&
	[|\uparrow \rangle, |\downarrow \rangle, |\downarrow \rangle] = |\uparrow \rangle\,,\\
	& [|\downarrow \rangle,|\downarrow \rangle,|\downarrow \rangle] = |\downarrow \rangle\, , && [|\downarrow \rangle,|\uparrow \rangle, |\uparrow \rangle ] =  |\downarrow \rangle \,.
	\end{align*}
	All other vector ternary products are equal to the zero vector $\mathbf{0}\in \mathcal{H}$. Using the linearity and conjugate linearity one can deduce the vector ternary product for arbitrary vectors (not necessarily normalised). For example
	$$[a |\uparrow \rangle,b|\uparrow \rangle, c|\uparrow \rangle + d|\downarrow \rangle]  = a \bar{b}c|\uparrow \rangle \,, $$
	with $a,b,c$ and $d \in \C$. 
\end{example}
\begin{example}
	The orthonormal basis of states for the one-dimensional harmonic oscillator is countably infinite as each basis vector is labelled by $n \in \mathbb{N}$ (including zero). The vector ternary product can be written in this natural basis (and then using linearity and conjugate linearity to deduce the product of arbitrary vectors) as 
	$$[| n_1 \rangle , |n_2 \rangle, | n_3 \rangle] = | n_1\rangle \, \delta_{n_2 n_3}\,.$$
\end{example}
\begin{remark}
	All quantum systems with a finite or countably infinite number of states, e.g., the hydrogen atom, have a vector ternary product that can easily be expressed in a similar way to the previous example.
\end{remark}
Recall that a linear map  $\varphi : \mathcal{H} \rightarrow \mathcal{H}'$ between Hilbert spaces is said to be \emph{bounded} if there exists some $r > 0$ such that $|| \, \varphi|\psi \rangle   \, ||' =  r \,|| \, |\psi \rangle   \, || $. It is a well-known result that boundedness implies continuity of a linear map and vice versa.
A \emph{bounded linear isometry} is a bounded linear map  $\varphi : \mathcal{H} \rightarrow \mathcal{H}'$ such that $\varphi^\dag \varphi = \Id_\mathcal{H}$.  
\begin{proposition}
	Let $\mathcal{H}$ and $\mathcal{H}'$ be Hilbert spaces and let  $\varphi : \mathcal{H} \rightarrow \mathcal{H}'$ be a  bounded linear isometry. Then $\varphi$ is morphism of semiheaps
	$$\varphi : (\mathcal{H}, [-,-,-]) \longrightarrow (\mathcal{H}',[-,-,-]')\,. $$
\end{proposition}
\begin{proof}
	Directly, using $\C$-linearity and the condition that the bounded linear map be an isometry, we observe that
	$$\varphi [|\psi_1 \rangle , |\psi_2 \rangle , |\psi_3\rangle] = \varphi\big ( | \psi_1 \rangle \langle \psi_2 | \psi_3 \rangle \big) = \varphi\big ( | \psi_1 \rangle\big) \langle \psi_2 | \psi_3 \rangle = \varphi\big ( | \psi_1 \rangle\big) \langle \psi_2 | \varphi^\dag \varphi|\psi_3 \rangle =  [\varphi|\psi_1 \rangle , \varphi|\psi_2 \rangle , \varphi|\psi_3\rangle]'\,.$$	
\end{proof}
\begin{remark}
	If we consider bounded linear maps that are not isometries, then we will not, in general, have a homomorphism of the relevant semiheaps.
\end{remark}
As we are considering linear maps, it is clear that bounded linear isometries are also ternary algebra homomorphisms. \par 
Unitary operators, i.e., bounded operators such that $U^\dag U =  U U^\dag = \Id_{\mathcal{H}}$, form a group, $\mathcal{U}(\mathcal{H})$,  and their action on $\mathcal{H}$ are isometries. In particular, the action $\rho_U : \mathcal{H} \rightarrow \mathcal{H}$ is $|\psi \rangle \mapsto U |\psi\rangle$ for arbitrary $U \in \mathcal{U}(\mathcal{H})$. We then have the following corollary. 
\begin{corollary}
	Let $\mathcal{U}(\mathcal{H})$ be the group of unitary operators on a Hilbert space $\mathcal{H}$. Furthermore, let $(\mathcal{H},  [-,-,-])$ be the associated semiheap. Then the action on $\mathcal{U}(\mathcal{H})$ on $\mathcal{H}$ is a semiheap isomorphism and so an isomorphism of ternary algebras.
\end{corollary}
Symmetries in quantum mechanics are usually understood as \emph{projective representations} of some group $G$. That is, we have  a map
$$U : G \longrightarrow \mathcal{U}(\mathcal{H} )\, , $$
such that $U(g_1)U(g_2) = \omega(g_1, g_2) \, U(g_1, g_2)$, with $\omega : G \times G \rightarrow U(1) := \{z \in \C, ~~ | ~ |z| =1 \}$, being referred to as the \emph{Schur factor}. Associativity implies that $\omega(g_1,g_2)\omega(g_1 g_2, g_3) = \omega(g_1, g_2 g_3) \omega(g_2, g_3)$. Assuming that  $U(e) = \Id_{\mathcal{H}}$ (as standard) implies that $\omega(e,e) = 1$. One can also deduce that $\omega(g,e) = \omega(e,g)=1$ and $\omega(g,g^{-1}) = \omega(g^{-1},g)$. If $\omega(g_1, g_2) = 1$ for all $g_1, g_2 \in G$, then we have a \emph{unitary representation}. Wigner's theorem (see \cite{Wigner:1959}) tells us that symmetries in quantum mechanics act via either projective or unitary representations.  We thus, in general, have an ``action up to a factor'' $\rho_U(-): G \times \mathcal{H} \rightarrow \mathcal{H}$  given by $(g,| \psi \rangle ) \mapsto U(g) | \psi \rangle$.
\begin{corollary}
	Let $\mathcal{U}(\mathcal{H})$ be the group of unitary operators on a Hilbert space $\mathcal{H}$ and let $U : G \longrightarrow \mathcal{U}(\mathcal{H} )$ be a projective representation. Furthermore, let $(\mathcal{H},  [-,-,-])$ be the associated semiheap. Then, for any $g \in G$, $\rho_U(g) :\mathcal{H} \rightarrow \mathcal{H}$ is a semiheap homomorphism and so a homomorphism of ternary algebras.
\end{corollary}
\begin{remark}
	The dual of a Hilbert space also comes with the canonical structure of a semiheap and ternary algebra by defining $[ \langle \psi_3 |  , \langle \psi_2 |, \langle \psi_1 |] := \langle \psi_3 | \psi_2 \rangle \langle \psi_1 |$. By construction we have $[|\psi_1 \rangle, |\psi_2\rangle , | \psi_3\rangle]^\dag = [ \langle \psi_3 |  , \langle \psi_2 |, \langle \psi_1 |] $. Note that although we can canonically identify a Hilbert space and its dual, we consider them as distinct spaces. 
\end{remark}
The vector ternary product can be extended to direct sums of Hilbert spaces as follows.  Recall that the (orthogonal) direct sum $\mathcal{H} =  \mathcal{H}_1 \oplus \mathcal{H}_2$ comes equipped with an inner product given by 
$$\big (  |\psi_1 \rangle + |\phi_1 \rangle, \,  |\psi_2 \rangle + |\phi_2 \rangle \big) \longmapsto  \langle \psi_1 |\psi_2 \rangle + \langle \phi_1 |\phi_2 \rangle\,. $$
Then, the vector ternary product is given by
$$[ |\psi_1 \rangle + |\phi_1 \rangle, |\psi_2 \rangle + |\phi_2 \rangle,  |\psi_3 \rangle + |\phi_3 \rangle] :=   | \psi_1 \rangle \langle \psi_2 | \psi_3 \rangle +  | \phi_1 \rangle \langle \phi_2 | \phi_3 \rangle = [ |\psi_1 \rangle , |\psi_2 \rangle ,  |\psi_3 \rangle ] + [  |\phi_1 \rangle, |\phi_2 \rangle    + |\phi_3 \rangle]\,.$$
This construction extends to the orthogonal direct sum of any finite number of Hilbert spaces.
\begin{example}
	In supersymmetric quantum mechanics, the relevant Hilbert space is the (orthogonal) direct sum on the bosonic sector $\mathcal{H}_0$ and the fermionic sector $\mathcal{H}_1$, i.e., $\mathcal{H} =  \mathcal{H}_0 \oplus \mathcal{H}_1$.  Of course, being orthogonal implies that linear combinations of bosonic and fermionic states cannot be physically realised. Nonetheless, we can still consider the vector ternary product on the direct sum as the sum of two vector ternary products on each sector.
\end{example}
Similarly, the vector ternary product can be extended to the tensor product of Hilbert spaces. We denote the (completed)  tensor product as $\mathcal{H} = \mathcal{H}_1 \otimes \mathcal{H}_2$. We remark that composite quantum systems are always described via the tensor products of their components.  Basic elements of $\mathcal{H}$ are pairs which, as standard,  we write as $| \psi \rangle \otimes |\phi\rangle$. The inner product (used for the completion) is, on basic elements, given by 
$$\big ( | \psi_1 \rangle \otimes |\phi_1\rangle, \,  | \psi_2 \rangle \otimes |\phi_2\rangle \big) \longmapsto \langle \psi_1| \psi_2 \rangle \, \langle \phi_1| \phi_2 \rangle\,,$$
which is then extended via linearity. The vector ternary product (on basic elements) is given by
\begin{align*}
[| \psi_1 \rangle \otimes |\phi_1\rangle,  | \psi_2 \rangle \otimes |\phi_2\rangle, | \psi_3 \rangle \otimes |\phi_3\rangle   ] & : = (| \psi_1 \rangle \otimes |\phi_1 \rangle)\langle \psi_2|\psi_3\rangle  \langle \phi_2|\phi_3\rangle\\
& =  | \psi_1 \rangle \langle \psi_2|\psi_3\rangle \otimes |\phi_1 \rangle  \langle \phi_2|\phi_3\rangle\\
& = [| \psi_1 \rangle ,  | \psi_2 \rangle , | \psi_3 \rangle    ] \otimes [| \phi_1 \rangle ,  | \phi_2 \rangle , | \phi_3 \rangle    ] \,.
\end{align*}
We observe that quite as expected, the vector ternary product on a tensor product of Hilbert spaces is the tensor product of the vector ternary products.  This construction then generalises to any finite tensor product of Hilbert spaces.
\subsection{Bounded linear operators and their ternary algebra}
We will denote the $*$-algebra of bounded (so, continuous) operators on $\mathcal{H}$ by $\mathcal{B}(\mathcal{H})$. Following our previous notation, we may also mean by  $\mathcal{B}(\mathcal{H})$ just the set of bounded linear operators, the context should be clear. If we want to consider just the vector space structure then we will write $( \mathcal{B}(\mathcal{H}), +)$.
\begin{definition}\label{Def:ConTerProdOp}
	Let $\mathcal{H}$ be a Hilbert space and let  $\mathcal{B}(\mathcal{H})$ be the  the $*$-algebra of bounded operators on $\mathcal{H}$. The \emph{operator ternary product}  $[-,-,-]: \mathcal{B}(\mathcal{H}) \times\mathcal{B}(\mathcal{H}) \times \mathcal{B}(\mathcal{H}) \longrightarrow \mathcal{B}(\mathcal{H}) $  is defined as 
	$$[A_1, A_2, A_3] :=A_1 A_2^\dag A_3\,.$$
\end{definition}
\begin{remark}
	We focus on bounded linear operators to avoid mathematical subtleties with taking adjoints and forming algebras under composition. 
\end{remark}
\begin{remark}
	The ternary product of bounded operators is closely related to the notion of a ternary ring of operators between Hilbert spaces as first introduced by  Hestenes \cite{Hestenes:1962} and extended to the $C^*$-algebra case by Zettl \cite{Zettl:1983}.
\end{remark}
\begin{proposition}\label{Prop:OpTernProd}
	The operator ternary product on $\mathcal{B}(\mathcal{H})$, see Definition \ref{Def:ConTerProdOp}, 
	\begin{enumerate}
		\item is linear in the first and third arguments, conjugate linear in the second argument, and  \label{Prop:LinOpTer}
		\item satisfies the para-associative law, or in other words, $(\mathcal{B}(\mathcal{H}), [-,-,-])$ is a semiheap.  \label{Prop:SimHepOpTer}
	\end{enumerate}	
\end{proposition}
\begin{proof}
	Part \eqref{Prop:LinOpTer} is clear from the definition. Part \eqref{Prop:SimHepOpTer} follows from a direct calculation.  Specifically,
	\begin{align*}
	& \big [ [A_1 ,A_2, A_3], A_4, A_5 \big] = A_1 A_2^\dag A_3 A_4^\dag A_5\,, \\
	& \big [ A_1 , [ A_4, A_3, A_2], A_5 \big] =  [A_1, A_4 A_3^\dag A_2, A_5 ] = A_1(A_4 A_3^\dag A_2)^\dag A_5 =  A_1 A_2^\dag A_3 A_4^\dag A_5\,, \\
	& \big [ A_1 ,A_2, [A_3, A_4, A_5] \big] =   A_1 A_2^\dag A_3 A_4^\dag A_5\,.
	\end{align*}
\end{proof}
\begin{definition}
	Let   $\mathcal{B}(\mathcal{H})$ be the  the $*$-algebra of bounded operators on a Hilbert space $\mathcal{H}$. Then the ternary algebra  $(\mathcal{B}(\mathcal{H}), +, [-,-,-])$ defined via Proposition \ref{Prop:OpTernProd} is referred to as the \emph{operator ternary algebra}.
\end{definition}
\begin{example}
	Considering the complex line, it is clear that $ \mathcal{B}(\mathbb{C}) = \textnormal{Mat}_{1\times 1}(\mathbb{C}) = \mathbb{C}$.  Thus, the operator and vector ternary products are identical, see Example \ref{Exp:ComplexLine}.  
\end{example}
\begin{example}
	Continuing Example \ref{Exp:SpinSpace}, as the Hilbert space is isomorphic to $\C^2$, it is clear that $\mathcal{B}(\C^2) \cong \textnormal{Mat}_{2\times 2}(\mathbb{C})$. To set some notation, we denote the components of a matrix with respect to the standard basis as $A_i^{\,\,j}$ and the components of the Hermitian conjugate as $\bar{A}^{j}_{\,\, i}$. Then the components of the operator ternary product are
	$$[A,B,C]_i ^{\,\,j} = A_i^{\,\,k} \bar{B}^{l}_{\,\, k}C_l^{\,\,j}\,. $$  
	The operator ternary product for $\C^n$ $(n \in \mathbb{N})$ is of the above from.
\end{example}
As mentioned earlier, unitary operators, i.e., bounded operators such that $U^\dag U =  U U^\dag = \Id_{\mathcal{H}}$, form a group. Because we have the structure of a group and $U^{-1} = U^\dag$, we have the following corollary. Alternatively, one needs only check the Mal’cev identities, and in this case, it is obvious they hold.
\begin{corollary}
	The group of unitary operators $\mathcal{U}\big ( \mathcal{H}\big)$ on a Hilbert space $\mathcal{H}$ is a heap under the operator ternary product.
\end{corollary}
As standard, we will denote the commutator of bounded operators as $[A_1, A_2] := A_1 A_2 - A_2 A_1$, for arbitrary $A_1$ and $A_2 \in \mathcal{B}(\mathcal{H})$.  We remind the reader that $[A_1, A_2]^\dag = - [A_1^\dag, A_2^\dag]$, and that we can cast the Jacobi identity into the Jacobi--Leibniz form
\begin{equation}\label{Eqn:JacIdeCom}
[A_1, [A_2, A_3]] = [[A_1, A_2] , A_3] + [A_2, [A_1, A_3]]\,.
\end{equation}
\begin{proposition}\label{Prop:GenJacLeiRule}
	The following identity holds for the operator ternary product on $\mathcal{B}(\mathcal{H})$, see Definition \ref{Def:ConTerProdOp},
	$$\big [ A_1, [A_2, A_3 ,A_4] \big] = \big [ [A_1, A_2], A_3 , A_4  \big ] - \big[   A_2,[A_1^\dag,A_3] , A_4\big ] + \big [A_2, A_3 , [A_1, A_4]\big]\,,$$
	for all $A_1, A_2, A_3$ and $A_4 \in\mathcal{B}(\mathcal{H}) $.
\end{proposition}
\begin{proof}
	Directly we observe that
	\begin{align*}
	\big[ A_1, [A_2, A_3 , A_4]\big] & = A_1 A_2 A_3^\dag A_4 - A_2  A_3^\dag A_4 A_1 \\
	& = A_1 A_2 A_3^\dag A_4 - A_2  A_3^\dag A_4 A_1 - A_2 A_1 A_3^\dag A_4\\
	& + A_2 A_1 A_3^\dag A_4 - A_2 A_3^\dag A_1 A_4 +  A_2 A_3^\dag A_1 A_4\\
	& = \big [ [A_1, A_2], A_3 , A_4  \big ] - \big[   A_2,[A_1^\dag,A_3] , A_4\big ] + \big [A_2, A_3 , [A_1, A_4]\big]\,.
	\end{align*}
\end{proof}
We interpret Proposition \ref{Prop:GenJacLeiRule} as a generalised version of the Leibniz rule for the commutator over the ternary product, and this should be compared with \eqref{Eqn:JacIdeCom}.  We make the following observation.
\begin{corollary}
	If $A_1 \in\mathcal{B}(\mathcal{H}) $ is self-adjoint, i.e., $A_1^\dag =  A_1$, then $[\rmi A_1,-]$ is a derivation over the operator ternary product on $\mathcal{B}(\mathcal{H}) $, i.e.,
	$$\big [ \rmi A_1, [A_2, A_3 ,A_4] \big] = \big [ [\rmi A_1, A_2], A_3 , A_4  \big ] + \big[   A_2,[\rmi A_1 ,A_3] , A_4\big ] + \big [A_2, A_3 ,  [\rmi A_1, A_4]\big]\,.$$
\end{corollary}
The unitary group $\mathcal{U}(\mathcal{H})$ acts on $\mathcal{B}(\mathcal{H})$ via similarity transformations. That is, $\rho_U : \mathcal{B}(\mathcal{H}) \rightarrow \mathcal{B}(\mathcal{H})$ is given by $A \mapsto U^\dag A U$, for arbitrary $U \in\mathcal{U}(\mathcal{H}) $. We then have the following proposition.
\begin{proposition}
	Let $\big (\mathcal{B}(\mathcal{H}), [-,-,-]\big)$ be the semiheap associated with bounded linear operators on a Hilbert space $\mathcal{H}$. Then,  the action of the unitary group  $\mathcal{U}(\mathcal{H})$ on $\mathcal{B}(\mathcal{H})$ is a semiheap homomorphism.
\end{proposition}
\begin{proof}
	The proposition is proved via direct calculation. Specifically,
	\begin{align*}
	\rho_U \big( [A_1, A_2,A_3]\big) & = U^\dag  [A_1, A_2,A_3] U = U^\dag  A_1 A_2^\dag A_3 U 
	= U^\dag A_1 U \big( U^\dag A_2^\dag U \big)U^\dag A_3 U \\
	&= [U^\dag A_1 U, U^\dag A_2 U, U^\dag A_3 U ]
	=  [\rho_U(A_1),\rho_U (A_2), \rho_U(A_3)]\,.
	\end{align*}
\end{proof}
\begin{corollary}
	Let $\mathcal{U}(\mathcal{H})$ be the group of unitary operators on a Hilbert space $\mathcal{H}$ and let $U : G \longrightarrow \mathcal{U}(\mathcal{H} )$ be a projective representation. Furthermore, let $\big (\mathcal{B}(\mathcal{H}), [-,-,-]\big)$ be the semiheap associated with bounded linear operators. Then, for any $g \in G$, $\rho_U(g) :\mathcal{\mathcal{H}} \rightarrow \mathcal{\mathcal{H}}$ is a semiheap homomorphism and so a homomorphism of ternary algebras.
\end{corollary}
Note that $[A_1, A_2, A_3]^\dag = [A_3^\dag, A_2^\dag, A_1^\dag]$ and so the operator ternary product is well-behaved with respect to taking adjoints. We denote the set of bounded self-adjoint operators, so the bounded observables, as $\mathcal{B}_s(\mathcal{H})$. Two operators $A$ and $B\in \mathcal{B}_s(\mathcal{H})$ are said to be \emph{compatible bounded observables} if they commute, i.e., $AB = BA$.  A \emph{compatible set of bounded observables} is a subset of $\mathcal{B}_s(\mathcal{H})$ such that all elements are pairwise compatible, that is, they pairwise commute. Naturally, a sub-semiheap of a semiheap is a subset that is closed with respect to the semiheap operation.
\begin{proposition}
	Let $\mathcal{B}_s(\mathcal{H})$ be the set of bounded observables on a Hilbert space $\mathcal{H}$. Then any compatible set of bounded observables is closed with respect to the operator ternary product. In other words, any set of compatible bounded observables forms a sub-semiheap of $\big (\mathcal{B}(\mathcal{H}), [-,-,-] \big)$.
\end{proposition}
\begin{proof}
	Consider three arbitrary (not necessarily distinct) bounded observables $A, B$ and $C \in \mathcal{B}_s(\mathcal{H})$. Then directly
	$$[A,B,C]^\dag = C^\dag B A^\dag = CBA = [C,B,A]\,.$$
	Upon the assumption these bounded observables pairwise commute we see that $CBA = ABC$ and so $[A,B,C]^\dag = [A,B,C]$ as required.
\end{proof}
\subsection{Distributivity of operators and derivations}
From the definition of the vector ternary product on a Hilbert space $\mathcal{H}$, see Definition \ref{Def:TerProdVec}, we have the following ``distributive law'',
\begin{equation}\label{Eqn:DisRuleOps}
A[| \psi_1 \rangle  , | \psi_2 \rangle, | \psi_3 \rangle] = [A| \psi_1 \rangle  , | \psi_2 \rangle, | \psi_3 \rangle]\,,
\end{equation}
for all $| \psi_1 \rangle,| \psi_2 \rangle $ and $| \psi_3 \rangle \in \mathcal{H}$, and all $A\in \mathcal{B}(\mathcal{H})$. The following was first, to our knowledge, uncovered by Kerner \cite{Kerner:2008}. Let us suppose the Hilbert space in question is finite or countable infinite. Furthermore, let us fix an orthonormal basis $\{ | n \rangle\}_{n \in \mathbb{N}}$. With respect to this fixed basis, any vector and operator can be written as
$$| \psi \rangle  =  \sum_{n =1}^\infty c_n \, |n\rangle\,, \qquad A =  \sum_{l,m}^\infty a_{ml}\, |l\rangle \langle m |\,. $$ 
Then, combining the two above expressions
\begin{equation}\label{Eqn:OpsVect}
A |\psi \rangle  = \sum_{n,m,l =1}^\infty c_n a_{ml} \, |l \rangle \langle m | n \rangle  =  \sum_{n,m,l =1}^\infty c_n a_{ml} \, [|l \rangle, |m  \rangle,|n \rangle] \,.
\end{equation}
By employing semiheaps and para-associative ternary algebras, we have a unification scheme in which vectors (states) and operators (observables) are treated as the same. It is linear combinations of triplets of vectors that are central to the theory rather than separately vectors and operators. \par 
The distributivity law  \eqref{Eqn:DisRuleOps} can be written in the form of a generalised Leibniz rule, and this should directly be compared with Proposition \ref{Prop:GenJacLeiRule}. 
\begin{proposition}\label{Prop:GenLeiRul}
	Let $\mathcal{H}$ be a Hilbert space and let $[-,-,-]$ be the associated vector ternary product. Then any bounded linear operator $A \in \mathcal{B}(\mathcal{H})$ satisfies a generalised ternary Leibniz rule
	$$A[| \psi_1 \rangle  , | \psi_2 \rangle, | \psi_3 \rangle] = [A| \psi_1 \rangle  , | \psi_2 \rangle, | \psi_3 \rangle] -  [| \psi_1 \rangle  , A^\dag| \psi_2 \rangle, | \psi_3 \rangle] +   [| \psi_1 \rangle  ,| \psi_2 \rangle,A | \psi_3 \rangle]\,,$$
	for all $| \psi_1 \rangle, | \psi_2 \rangle$ and $| \psi_3 \rangle \in \mathcal{H}$.
\end{proposition}
\begin{proof}
	In light of \eqref{Eqn:DisRuleOps},  we require that $-  [| \psi_1 \rangle  , A^\dag| \psi_2 \rangle, | \psi_3 \rangle] +   [| \psi_1 \rangle  ,| \psi_2 \rangle,A | \psi_3 \rangle] =0$. However, this is the case for any bounded operator $A$ as,  directly from Definition \ref{Def:ConTerProdOp}, $ [| \psi_1 \rangle  , A^\dag| \psi_2 \rangle, | \psi_3 \rangle] = | \psi_1 \rangle \langle \psi_2 | A |  \psi_3 \rangle = [| \psi_1 \rangle  ,| \psi_2 \rangle,A | \psi_3 \rangle]  $. 
\end{proof}
\begin{definition}
	Let $\mathcal{H}$ be a Hilbert space and let $[-,-,-]$ be its associated vector ternary product. A bounded linear operator $D \in \mathcal{B}(\mathcal{H})$ is said to be a \emph{derivation of the vector ternary product} on $\mathcal{H}$  if it satisfies the ternary Leibniz rule
	$$D[| \psi_1 \rangle  , | \psi_2 \rangle, | \psi_3 \rangle] = [D| \psi_1 \rangle  , | \psi_2 \rangle, | \psi_3 \rangle] + [| \psi_1 \rangle  , D | \psi_2 \rangle, | \psi_3 \rangle] + [| \psi_1 \rangle  , | \psi_2 \rangle,D | \psi_3 \rangle] \,,$$
	for all $| \psi_1 \rangle,| \psi_2 \rangle $ and $| \psi_3 \rangle \in \mathcal{H}$.
\end{definition}
There is a one-to-one correspondence between anti-self-adjoint and self-adjoint operators given by multiplication by $\rmi = \sqrt{-1}$.  Specifically, if $A$ is anti-self-adjoint, then $\rmi A$ is self-adjoint, i.e., $(\rmi A)^\dag =\rmi A $. Conversely,  if $B$ is self-adjoint, then $\rmi B$ is anti-self-adjoint, i.e., $(\rmi B)^\dag = - \rmi B $. The following proposition appears in \cite[Section 6]{Kerner:2008}.
\begin{proposition}
	There is a one-to-one correspondence between the set of derivations of the vector ternary product on $\mathcal{H}$ and the set of bounded observables $\mathcal{B}_s(\mathcal{H})$.
\end{proposition}
\begin{proof}
	In light of \eqref{Eqn:DisRuleOps}, it is clear that $\langle \psi_2 | D^\dag |\psi_3 \rangle +\langle \psi_2 | D|\psi_3 \rangle=0 $ if a bounded linear operator is a derivation. Thus, as the vectors in $\mathcal{H}$ are arbitrary, $D^\dag  =  - D$. That is, $D$ must be anti-self-adjoint.  We can always find a unique self-adjoint operator $A \in \mathcal{B}(\mathcal{H})$ such that $D = \rmi A$. Conversely, any self-adjoint operator $A$ corresponds to an anti-self-adjoint operator $\rmi A =D$. 
\end{proof}
\begin{proposition}
	Derivations of the vector ternary product on a Hilbert space $\mathcal{H}$ are closed under the commutator.
\end{proposition}
\begin{proof}
	If $D_1$ and $D_2$ are anti-self-adjoint operators, then $[D_1, D_2]^\dag = - [D_1, D_2]$, i.e., the commutator is also anti-self-adjoint.
\end{proof}
It is clear that the linear combination $a\, D_1 + b \, D_2$ is also anti-self-adjoint for $a$ and $b\in \R$. Note, rather obviously, this is not the case for linear combinations with complex coefficients with non-zero imaginary parts. We then have the following observation.
\begin{corollary}
	Derivations of the vector ternary product on a Hilbert space $\mathcal{H}$ form a real Lie algebra with respect to the commutator bracket.
\end{corollary}
\section{Concluding remarks}
In this paper, we have re-examined the semiheaps and associated para-associative algebras that are naturally present in the theoretical setup of quantum mechanics. In particular, their symmetries and generalised derivations have been studied, this has, to the author's knowledge, not been explored before. \par 
Interestingly, semiheaps allow one to treat vectors and operators as non-distinct objects (see \eqref{Eqn:OpsVect}). The action of an operator on a state is replaced by a linear combination of triplets of states fed into the vector ternary product. As far as we know, this observation, first made by Kerner, has not been exploited in quantum mechanics.\par 
Note that Proposition \ref{Prop:GenJacLeiRule} and Proposition \ref{Prop:GenLeiRul} suggest that for a para-associative ternary product, the generalisation of Leibniz rule should be of the form $D[a,b,c] = [Da, b, c] - [a, D^\dag b , c] + [a,b, D c]$ for all elements $a,b$ and $c$. This is in contrast to the obvious direct generalisation of the Leibniz rule.  In particular, we note that there is a linear combination of objects of the form `$+\, - \, +$' and that this is a sign that a heap operation is at play here.  To the author's mind, this modification of the Leibniz rule deserves further investigation in the general setting of ternary algebras and not just those presented here. \par
In conclusion, quantum mechanics has provided much inspiration for the study of operator algebras and noncommutative structures. Similarly, quantum mechanics provides impetus for the investigation of ternary algebras and non-associative structures.

\section*{Acknowledgements}
The author thanks Steven Duplij for his encouragement to complete this work. A special thank you goes to Tomasz Brzezinski for introducing the author to heaps and related structures.


\begin{thebibliography}{10}
\begin{small}

\bibitem{Abramov:2009}
Viktor Abramov, Richard  Kerner, Olga Liivapuu \& Sergei Shitov, 
Algebras with ternary law of composition and their realization by cubic matrices,
\href{https://doi.org/10.4303/jglta/S090201}{\emph{J. Gen. Lie Theory Appl.}} \textbf{3} (2009), no. 2, 77--94. 


\bibitem{Azcarraga:2010}
Jose A. de  Azcárraga \& Jose M. Izquierdo, n-ary algebras: a review with applications, \href{https://doi.org/10.1088/1751-8113/43/29/293001}{\emph{J. Phys. A: Math. Theor.} }\textbf{43} (2010) 293001.
	
	
\bibitem{Bagger:2007}
Jonathan Bagger \& Neil Lambert,
Modeling multiple M2-branes,
\href{https://doi.org/10.1103/PhysRevD.75.045020}{\emph{Phys. Rev. D}} \textbf{75} (2007), no. 4, 045020, 7 pp.	
	
\bibitem{Baer:1929}
 Reinhold Baer, Zur Einführung des Scharbegriffs, \href{https://doi.org/10.1515/crll.1929.160.199}{\emph{J. Reine Angew. Math.}} \textbf{160} (1929), 199--207.
 %
 \bibitem{Baxter:1972}
 Rodney J. Baxter, 
Partition function of the eight-vertex lattice model,
\href{https://doi.org/10.1016/0003-4916(72)90335-1}{\emph{Ann. Physics}} \textbf{70} (1972), 193--228.
 %
 \bibitem{Bazunova:2004}
 Nadezda Bazunova, Andrzej Borowiec \& Richard Kerner,
 Universal differential calculus on ternary algebras,
 \href{https://doi.org/10.1023/B:MATH.0000035030.12929.cc}{\emph{Lett. Math. Phys.}} \textbf{67} (2004), no. 3, 195--206.
 %
 \bibitem{Brzezinski:2020}
 Tomasz Brzeziński, 
 Trusses: paragons, ideals and modules,
 \href{https://doi.org/10.1016/j.jpaa.2019.106258}{\emph{J. Pure Appl. Algebra}} \textbf{224} (2020), no. 6, 106258, 39 pp.
 %
 \bibitem{Gustavsson:2009}
 Andreas Gustavsson, 
Algebraic structures on parallel M2-branes,
\href{https://doi.org/10.1016/j.nuclphysb.2008.11.014}{\emph{Nuclear Phys. B} } \textbf{811} (2009), no. 1-2, 66--76. 
 
 %
\bibitem{Hestenes:1962}
Magnus R. Hestenes, 
A ternary algebra with applications to matrices and linear transformations,
\href{https://doi.org/10.1007/BF00253936}{\emph{Arch. Rational Mech. Anal.}} \textbf{11} (1962), 138--194.
	%
	\bibitem{Hollings:2017}
Christopher D. Hollings \&  Mark V. Lawson, 
\emph{Wagner's theory of generalised heaps}, \href{https://doi.org/10.1007/978-3-319-63621-4}{Springer, Cham,} 2017. xv+189 pp. ISBN: 978-3-319-63620-7.	
	%
\bibitem{Kerner:2008}
Richard Kerner,
Ternary and non-associative structures,
\href{https://doi.org/10.1142/S0219887808003326}{\emph{Int. J. Geom. Methods Mod. Phys.}} \textbf{5} (2008), no. 8, 1265--1294. 
%
\bibitem{Kerner:2018}
Richard Kerner, Ternary generalizations of graded algebras with some physical applications, \href{http://imar.ro/journals/Revue_Mathematique/pdfs/2018/2/4.pdf}{\emph{Rev. Roumaine Math. Pures Appl.} }\textbf{63} (2018), no. 2, 107--141.
%

\bibitem{Lada:1993}
Tom Lada \& Jim Stasheff, 
Introduction to SH Lie algebras for physicists,
\href{https://doi.org/10.1007/BF00671791}{\emph{Internat. J. Theoret. Phys.}} \textbf{32} (1993), no. 7, 1087--1103.
%
\bibitem{Michor:1996}
Peter W. Michor \&  Alexandre M. Vinogradov, n-ary Lie and associative algebras, \emph{Rend. Sem. Mat. Univ. Pol. Torino} \textbf{54}(4) (1996), 373-392.
%
\bibitem{Nambu:1973}
Yoichiro Nambu, 
Generalized Hamiltonian dynamics,
\href{https://doi.org/10.1103/PhysRevD.7.2405}{\emph{Phys. Rev. D}} (3) \textbf{7} (1973), 2405--2412.
%
\bibitem{Prufer:1924}
Heinz Pr\"{u}fer, 
Theorie der Abelschen Gruppen,
\href{https://doi.org/10.1007/BF01188079}{\emph{Math. Z.}} \textbf{20} (1924), no. 1, 165--187.
%
\bibitem{Wigner:1959}
Eugene P. Wigner, 
``Group theory and its application to the quantum mechanics of atomic spectra'',
expanded and improved ed., Pure and Applied Physics. Vol. 5 \emph{Academic Press, New York-London} 1959 xi+372 pp. 
%
\bibitem{Yang:1967}
Chen-Ning  Yang,
Some exact results for the many-body problem in one dimension with repulsive delta-function interaction,
\href{https://doi.org/10.1103/PhysRevLett.19.1312}{\emph{Phys. Rev. Lett.}} \textbf{19} (1967), 1312--1315. 
%
\bibitem{Zettl:1983}
Heinrich Zettl,  
A characterization of ternary rings of operators,
\href{https://doi.org/10.1016/0001-8708(83)90083-X}{\emph{Adv. in Math.}} \textbf{48} (1983), no. 2, 117--143.
\end{small}
\end{thebibliography}
\end{document}